\pdfoutput=1
\documentclass[conference,9pt]{IEEEtran}
\IEEEoverridecommandlockouts

\usepackage[T1]{fontenc}
\usepackage[utf8]{inputenc}
\usepackage{microtype}
\usepackage{amsmath,amssymb}
\usepackage{booktabs,tabularx,threeparttable,array}
\usepackage{graphicx}
\usepackage{xcolor}
\usepackage{listings}
\usepackage{tikz}
\usetikzlibrary{positioning,arrows.meta,calc}
\usepackage{balance}
\usepackage[numbers,sort&compress]{natbib}
\usepackage[hidelinks,breaklinks]{hyperref}
\hypersetup{pdftitle={Privacy-Friendly Cohort Determination: Sealed, CSP-Independent In-Browser ML Inference of Professional Segments for Identity-Less Advertising},pdfauthor={Om Shankar Tiwari, Navnit Shukla, Guanyu Wang, Akshay Jain}}
\usepackage{url}

\newcommand{\sys}{SIF}
\newcommand{\krr}{$k$-RR}
\newcommand{\lifatid}{\texttt{li\_fat\_id}}
\newcommand{\sharedid}{SharedID}
\newcommand{\Ltilde}{\tilde{L}}

\makeatletter
\newcommand{\linebreakand}{\end{@IEEEauthorhalign}\hfill\mbox{}\par\mbox{}\hfill\begin{@IEEEauthorhalign}}
\makeatother
\newtheorem{proposition}{Proposition}
\newtheorem{definition}{Definition}

\title{Privacy-Friendly Cohort Determination: Sealed, CSP-Independent In-Browser ML Inference of Professional Segments for Identity-Less Advertising}
\author{\IEEEauthorblockN{Om Shankar Tiwari}
\IEEEauthorblockA{Applied AI Technical Lead, Google\\\url{https://orcid.org/0009-0001-2247-0383}}
\and
\IEEEauthorblockN{Navnit Shukla}
\IEEEauthorblockA{Principal AI Architect, Snowflake Inc.\\\url{https://orcid.org/0009-0005-8801-3344}}
\linebreakand
\IEEEauthorblockN{Guanyu Wang}
\IEEEauthorblockA{Staff Software Engineer, TikTok}
\and
\IEEEauthorblockN{Akshay Jain}
\IEEEauthorblockA{Staff Software Engineer, LinkedIn}
\thanks{This work was carried out independently of the authors' employers and uses only public sources. The views are the authors' own.}}

\begin{document}
\maketitle

\begin{abstract}
B2B advertising targets a viewer's professional attributes (employer size and industry, function, seniority) and has obtained them by matching identities across sites. Safari and Firefox block third-party cookies, Google retired the Privacy Sandbox cohort APIs in 2025, and reverse-IP firmographics decay under remote work. We present \sys{} (Sealed Inference Frame), which infers coarse professional cohorts on the device and emits only a locally differentially private, taxonomy-coded label into the OpenRTB bid stream, with no cross-site identifier. It rests on a property of the web platform we make precise: a navigated cross-origin iframe is the only way third-party code obtains a policy it controls, so inference runs in WebAssembly even where the publisher's CSP forbids it, and a nested worker served with \texttt{default-src 'none'} gives the model no network. Even a malicious model leaks at most about 5 bits per site per week. Labels pass through a memoised $k$-ary randomised response keyed to the publisher's first-party identifier, which gives $\varepsilon$-local differential privacy, defeats averaging, and links requests no better than the identifier already sent. An org-conditional $k$-anonymity rule suppresses cells, more strictly on corporate networks than at home. Cohorts ride OpenRTB \texttt{user.data} in a LinkedIn-aligned taxonomy, and attribution uses LinkedIn's click-scoped \lifatid{} without bridging identities. We report a crawl of CSP deployment on 7,969 top sites and 431 B2B publishers, Heavy-Ad budgets, closed-form privacy--utility trade-offs, a re-identification simulation, and an assessment of which attributes are predictable at all: company type and size are, seniority largely is not. On-device is a design property, not a consent exemption.
\end{abstract}


\section{Introduction}
\label{sec:intro}

Business-to-business advertising wants to know something unusual. A consumer campaign asks whether a viewer likes running shoes. A B2B campaign asks whether the viewer is a director of engineering at a 5,000-person software company. The attributes that carry value (employer size, industry, job function, seniority) are stable and sparse, and for most of the web's history they were obtained by matching identity: a cookie seen on a publisher was joined to a profile held by a platform that knows the viewer's job. That join made the market work, and it is disappearing.

Three developments in 2025--2026 changed the setting. Chrome kept third-party cookies~\cite{google2025nextsteps}, but Safari~\cite{webkit2020cookies} and Firefox~\cite{mozilla2022tcp} block them by default, so identity matching is blind on a large, demographically skewed share of traffic. Google retired the browser-native alternatives (Topics, Protected Audience, Attribution Reporting, Private Aggregation, IP Protection) in October 2025~\cite{google2025retire,chromium2025topicsremove,chromestatus-pa-remove}, after academic work showed the cohort primitives were linkable~\cite{berke2022floc,jha2023topics,beugin2024interest,long2024protected} and adoption stayed low~\cite{grib2026rise,beugin2026lessons}. The industry's identity-free format, IAB Tech Lab's Seller Defined Audiences (SDA), carries self-attested labels with no verification, and a 2023 audit found 75\% of SDA signals failed the specification~\cite{iab2023sdaguide}. The B2B fallback, reverse-IP firmographics, decays under remote work: vendors themselves report that IP-only attribution identifies 5--20\% of visitors on corporate networks and none at home~\cite{demandbase2020wfh}.

This paper grew out of the first author's earlier work on privacy for LinkedIn's advertising products, where one question kept returning: what can B2B targeting keep once identity matching is gone? We ask whether the attributes buyers need can be inferred on the device, from signals a browser legitimately exposes, and released in a form that carries no identity and a bounded amount of information, and whether an ad-tech vendor could deploy such a system on third-party publishers in 2026 with no browser vendor's help. We answer with a system, \sys{} (Sealed Inference Frame), and an assessment of what it can and cannot predict. We contribute: (i) a policy-independent, zero-egress execution architecture (Section~\ref{sec:sif}), resting on the fact that a navigated cross-origin iframe is the only way third-party script obtains a policy container it controls, with a proof that even a malicious model can push at most about 5 bits per site per week through the emitted label; (ii) org-conditional inference with suppression (Section~\ref{sec:inference}), stricter on corporate networks than at home, the inverse of reverse-IP vendors; (iii) a memoised randomised response keyed to the publisher's first-party identifier (Section~\ref{sec:privacy}), giving $\varepsilon$-local differential privacy per value, no averaging, and no linking capability beyond the identifier already in the bid request; (iv) identity-free integration with OpenRTB and LinkedIn's click-scoped \lifatid{}, plus an agreement statistic that makes self-declared seller segments auditable (Section~\ref{sec:integration}); and (v) a crawl of CSP deployment against WebAssembly and third-party frames, a re-identification simulation, and a synthesis of what published evidence says each attribute is predictable from (Section~\ref{sec:eval}).

The machine-learning claim is deliberately narrow. Published evidence supports inferring organisation type for essentially all traffic and a named organisation for a minority of enterprise sessions, but not seniority from device signals at useful accuracy. We treat seniority and function as hypotheses, and argue that a weak classifier is, for privacy, a feature. Finally, ``on-device'' is a design property, not a legal category: releasing any derivation of device-read information requires consent even when the computation is local~\cite{edpb2023guidelines,pecr-reg6}. \sys{} is consent-gated and makes the consented processing minimal and verifiable (Section~\ref{sec:legal}).

\section{Background and Related Work}
\label{sec:background}

\textbf{Client-side private advertising} is an old idea~\cite{juels2001targeted}. Adnostic~\cite{toubiana2010adnostic}, RePriv~\cite{fredrikson2011repriv} and their successors (surveyed by Ullah et al.~\cite{ullah2020survey}) keep the profile on the client, but all target consumer interests. None integrates with the real-time bidding (RTB) protocol through which most display inventory is sold, and none handles attributes such as employer and seniority, which are quasi-identifiers in small populations.

\textbf{Browser cohort APIs.} FLoC's cohort identifier worked as a fingerprinting surface~\cite{rescorla2021floc} and re-identified more than 95\% of users within weeks~\cite{berke2022floc}. Topics replaced it with an on-device classifier, top-5 topics per weekly epoch, a 5\% uniformly random topic and per-caller witness filtering~\cite{topics-explainer}. It is the only shipped on-device classifier with randomised output for advertising and the closest precedent to our mechanism, and stable-interest users were re-identified across its epochs~\cite{thomson2023topics,jha2023topics,beugin2024interest}. Carey et al.\ gave a hypothesis-testing framework for re-identification from compact signals~\cite{carey2023reid}. Protected Audience's reporting channels permit cross-site linkage~\cite{long2024protected}. Google retired both, with removal at Chrome 152 for Topics and 153 for Protected Audience~\cite{google2025retire,chromium2025topicsremove,chromestatus-pa-remove}. For a vendor in 2026 there is no browser-mediated cohort primitive. Whatever runs must run as ordinary web content.

\textbf{Reverse-IP firmographics} map an IP address to an organisation using registry allocations, autonomous-system data and reverse DNS. Vendors' own accounts give the limits: deterministic attribution ``5--20\% of the time when employees are on their corporate networks'' and ``0\%'' at home~\cite{demandbase2020wfh}, company-level identification on 10--40\% of traffic in conservative benchmarks~\cite{leadfeeder-anon}, and higher ``match rates'' only with cookies and identity graphs~\cite{warmly-matchrates}. 94\% of cellular ASes deploy carrier-grade NAT~\cite{richter2016cgn}, shared SASE egress lets thousands of enterprises arrive from one cloud node~\cite{zscaler-transactions}, and RFC~8981 temporary IPv6 addresses randomise only the interface identifier, so enterprise prefix attribution survives~\cite{rfc8981}. AS-to-organisation mapping~\cite{cai2010as2org,caida-as2org} and AS owner classification~\cite{ziv2021asdb} give the feature engineering for organisation type rather than identity, which is the part of the signal we keep.

\textbf{Inferring professional attributes.} Tweet content predicts occupational class at 52.7\% over nine groups against a 34.4\% baseline~\cite{preotiuc2015occupation}, passive phone sensing carries occupational information~\cite{khaokaew2024workr}, and owning an Apple device was the strongest single predictor of top-versus-bottom income quartile in 2016, at 69\%~\cite{bertrand2018coming}. We know of no published evidence that seniority or function can be predicted from device, network and context signals alone.

\textbf{WebAssembly and CSP.} The Web Almanac reports CSP on 19--22\% of pages, \texttt{script-src} in about 17\% of policies, and \texttt{unsafe-eval} in 77--78\% of those~\cite{almanac2024security,almanac2025security}. WebAssembly compilation is gated by \texttt{script-src} unless \texttt{'unsafe-eval'} or \texttt{'wasm-unsafe-eval'} is present~\cite{csp3}. Browser deep-learning inference runs 5--17$\times$ slower than native on CPU, dominated by load and warm-up~\cite{ma2019moving,wang2024anatomizing}. Fingerprinting research quantifies the entropy of the device signals we use~\cite{eckersley2010unique,gomezboix2018hiding,laperdrix2020survey}. What has not been studied is how a third party obtains an execution context with its own policy on a hostile publisher, and what egress guarantees that context can satisfy.

\section{Setting, Threat Model and Goals}
\label{sec:threat}

A user visits a publisher that sells display inventory through header bidding (Prebid.js) into supply-side platforms (SSPs), which forward OpenRTB bid requests to demand-side platforms (DSPs). A vendor operates \sys{}: it ships the model, serves the bridge frame, and computes a coarse IP-context vector. A B2B DSP such as LinkedIn's Audience Network buys cohort-targeted impressions and measures conversions on advertiser sites where its tag is installed (Fig.~\ref{fig:context}). The assets are the device features $x_{\mathrm{dev}}$, the IP-context vector $x_{\mathrm{ip}}$, the model's true output $L$ (one cohort value per axis), the user's cross-site identity, and the vendor's model.

\begin{figure*}[t]
\centering
\begin{tikzpicture}[
  font=\scriptsize,
  box/.style={draw, rounded corners=2pt, align=center, inner sep=4pt, minimum height=0.8cm},
  srv/.style={box, fill=black!4},
  cli/.style={box, fill=blue!4, draw=blue!50!black},
  frame/.style={box, dashed, fill=white},
  flow/.style={-{Stealth[length=1.8mm]}, thick},
  never/.style={-{Stealth[length=1.8mm]}, thick, draw=red!70!black, dashed},
  lbl/.style={font=\scriptsize, align=center, fill=white, inner sep=1.5pt},
  warn/.style={font=\scriptsize, align=center, inner sep=1.5pt, text=red!70!black}
]
% vendor above T1
\node[srv, text width=3.4cm] (vendor) at (5.4,3.4) {\textbf{Vendor origin}\\ model CDN (hash-pinned)\\ IP-context service (public registry data)};
% browser on the publisher site
\node[cli, text width=7.2cm, minimum height=2.9cm, anchor=north west] (browser) at (0,1.5) {};
\node[anchor=north west, font=\scriptsize\bfseries] at (0.1,1.45) {Browser, top-level site = publisher};
\node[box, text width=2.9cm, anchor=north west, fill=white] (t0) at (0.25,0.95) {\textbf{T0} publisher page\\ Prebid RTD submodule\\ consent, \sharedid{}, context};
\node[frame, text width=3.3cm, anchor=north west] (t1) at (3.6,0.95) {\textbf{T1} vendor iframe\\ device signals, coin secret,\\ memo table, \krr{}\\[1pt] \textbf{T2} sealed worker (Wasm)};
\draw[flow] (t1.west) -- (t0.east);
\node[lbl, fill=blue!4] at (3.55,-0.95) {T1 $\to$ T0: the noised label only (\texttt{postMessage})};
% advertiser site
\node[cli, text width=4.6cm, minimum height=1.6cm, anchor=north west] (adv) at (8.6,1.5) {};
\node[anchor=north west, font=\scriptsize\bfseries] at (8.7,1.45) {Browser, top-level site = advertiser};
\node[box, text width=4.1cm, anchor=north west, fill=white] (tag) at (8.85,0.95) {Insight Tag: stores the \lifatid{} cookie,\\ runs \sys{} a second time};
% exchange row below
\node[srv, text width=2.4cm] (ssp) at (3.5,-3.9) {\textbf{SSP}\\ header-bidding exchange};
\node[srv, text width=3.8cm] (dsp) at (9.9,-3.9) {\textbf{B2B DSP}\\ bids on the cohort; attribution\\ $\lifatid{} \mapsto$ bid-time cohort;\\ reports with floors ($\ge 30$/segment)};
% flows
\draw[flow] (vendor.south) -- node[lbl, right] {model bytes, $x_{\mathrm{ip}}$ inline} (t1.north);
\draw[flow] (1.7,-1.4) -- (1.7,-3.9) -- (ssp.west);
\node[lbl, text width=3.1cm, anchor=west] at (1.85,-2.45) {OpenRTB bid request: \texttt{user.data} cohort + noise metadata; no \texttt{user.eids}; \sharedid{} only if the publisher already sends it};
\draw[flow, <->] (ssp.east) -- node[lbl, above] {bid request / bid} (dsp.west);
\draw[flow] ([xshift=-1.0cm]dsp.north) -- node[lbl, left] {click carries\\ \lifatid{}} ([xshift=-1.0cm]adv.south);
\draw[flow] ([xshift=1.0cm]adv.south) -- node[lbl, right] {conversion +\\ \lifatid{}} ([xshift=1.0cm]dsp.north);
% never
\draw[never] (t1.south) -- ++(0,-0.4);
\node[warn, text width=2.0cm, anchor=north] at (6.3,-1.6) {$x_{\mathrm{dev}}$, the fused $(x_{\mathrm{dev}},x_{\mathrm{ip}})$ and the model's raw output never leave T1/T2};
\draw[never, {Stealth[length=1.8mm]}-{Stealth[length=1.8mm]}] ($(browser.north east)+(0,-1.0)$) -- node[warn, above=1pt, font=\tiny] {never joined} ($(adv.north west)+(0,-1.0)$);
\node[warn, text width=3.9cm, anchor=west] at (13.5,0.7) {\sharedid{} (publisher site) and \lifatid{} (advertiser site) live in different storage partitions and are never joined};
\end{tikzpicture}
\caption{Parties and data flows. The only tier with network access to the vendor is T1, which receives model bytes and the coarse IP-context vector and emits a noised label to T0. The bid stream carries the cohort as OpenRTB \texttt{user.data} alongside whatever the publisher already sends. Attribution on the demand side uses the click-scoped \lifatid{} the DSP already has, and cohort-level counts are released under the platform's existing floors.}
\label{fig:context}
\end{figure*}
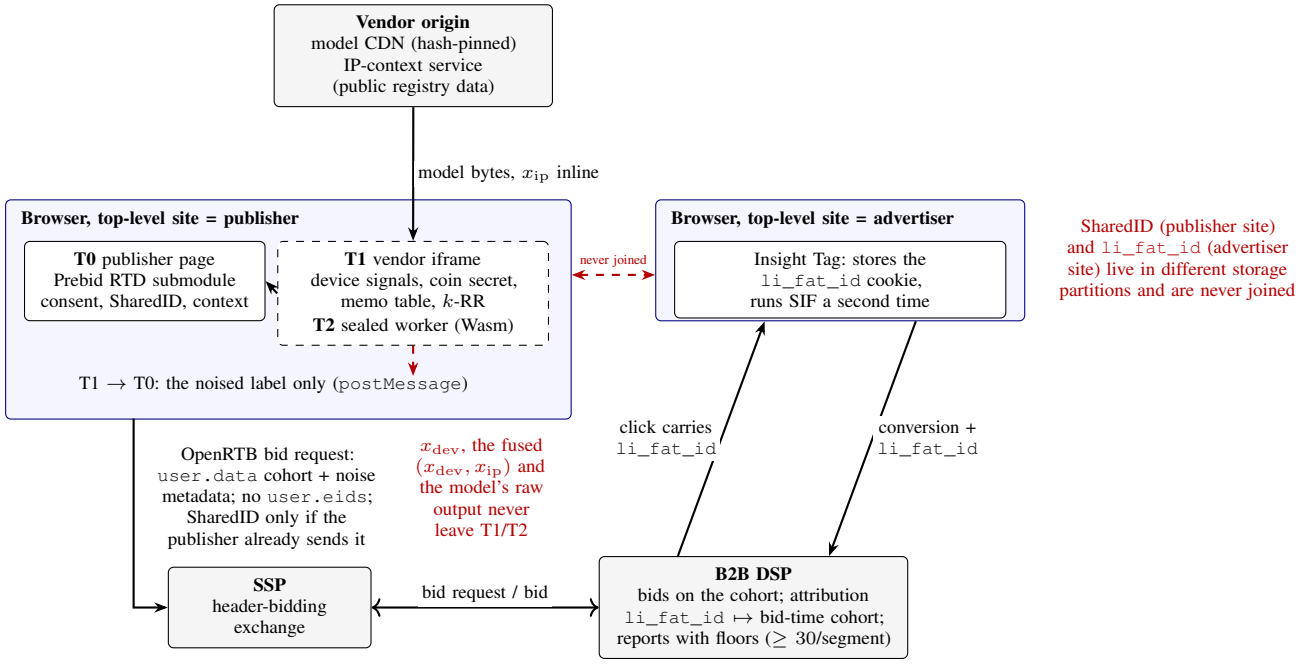

\textbf{Adversaries.} (A1) An honest-but-curious vendor, and its stronger form, a vendor whose weekly, opaque model updates turn malicious after publishers have integrated the SDK. (A2) SSPs, DSPs and any bid-stream observer, who see the label together with everything OpenRTB already carries: the IP (often full, sometimes truncated to /24 or /48~\cite{google-ab-openrtb}), user agent or UA hints, the publisher's first-party identifier if present, and page context. (A3) Colluding sites linking a user through the label. (A4) Users or bots spoofing features. (A5) A model thief. We do not defend against a malicious publisher, who already controls the page, or against a browser vendor.

\textbf{Goals.} (G1) \emph{No cross-site identifier}: nothing stable across top-level sites is created, stored or sent. (G2) \emph{Raw-signal confinement}: $x_{\mathrm{dev}}$ and the fused $(x_{\mathrm{dev}},x_{\mathrm{ip}})$ never leave the device. (G3) \emph{Bounded release}: each emitted value $\Ltilde$ is $\varepsilon$-locally differentially private with respect to $L$, and cells whose expected occupancy within the resolved organisation is below $k$ are suppressed. (G4) \emph{No incremental linkability}: conditioned on the publisher's identifier already in the request, the label gives A2 no additional ability to link two requests, and across sites its linkage advantage is bounded in closed form. (G5) \emph{Bounded covert channel}: a malicious model leaks at most $B$ bits per site per epoch. (G6) \emph{Small, pinned trusted base}: all code with network access is small, rarely changes and is integrity-pinned by the publisher. Non-goals are hiding the IP from web servers (IP Protection was retired~\cite{google2025retire}), operating without consent where consent is required (Section~\ref{sec:legal}), and protecting the model's confidentiality, which no browser mechanism can~\cite{sun2021mind}.

\section{The Sealed Inference Frame}
\label{sec:sif}

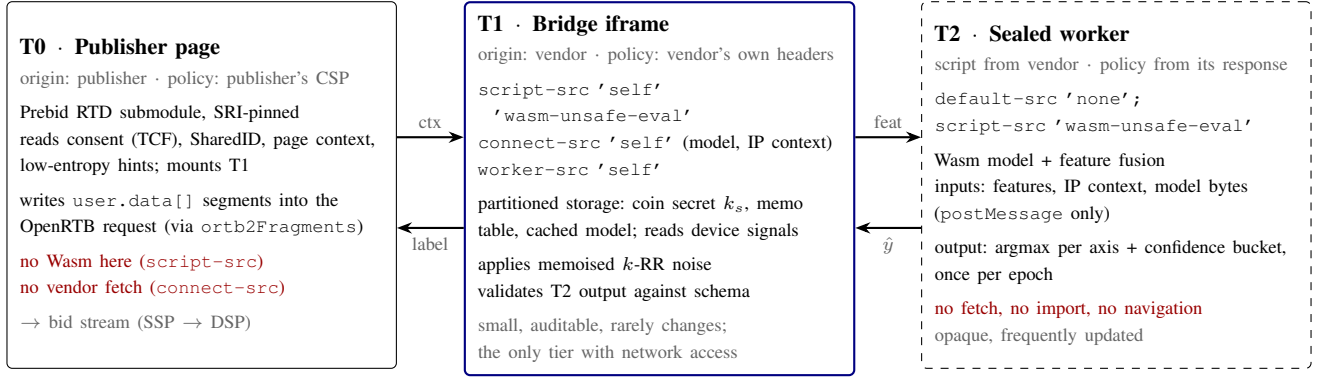
\begin{figure*}[t]
\centering
\begin{tikzpicture}[
  font=\small,
  tier/.style={draw, rounded corners=2pt, align=left, inner sep=5pt, text width=4.8cm, minimum height=4.85cm, anchor=north west},
  arr/.style={-{Stealth[length=2mm]}, thick}
]
\newcommand{\thdr}[1]{{\small\bfseries #1}}
\newcommand{\tsub}[1]{{\scriptsize\color{black!60}#1}}
\newcommand{\tbody}[1]{{\scriptsize #1}}
\newcommand{\twarn}[1]{{\scriptsize\color{red!60!black}#1}}
\node[tier] (t0) at (0,0) {%
  \thdr{T0 \,$\cdot$\, Publisher page}\\[1pt]
  \tsub{origin: publisher $\cdot$ policy: publisher's CSP}\\[3pt]
  \tbody{Prebid RTD submodule, SRI-pinned}\\
  \tbody{reads consent (TCF), \sharedid{}, page context,}\\
  \tbody{low-entropy hints; mounts T1}\\[3pt]
  \tbody{writes \texttt{user.data[]} segments into the}\\
  \tbody{OpenRTB request (via \texttt{ortb2Fragments})}\\[3pt]
  \twarn{no Wasm here (\texttt{script-src})}\\
  \twarn{no vendor fetch (\texttt{connect-src})}\\[3pt]
  \tsub{$\rightarrow$ bid stream (SSP $\rightarrow$ DSP)}
};
\node[tier, draw=blue!50!black, thick] (t1) at (6.05,0) {%
  \thdr{T1 \,$\cdot$\, Bridge iframe}\\[1pt]
  \tsub{origin: vendor $\cdot$ policy: vendor's own headers}\\[3pt]
  \tbody{\texttt{script-src 'self'}}\\
  \tbody{\texttt{\ \ 'wasm-unsafe-eval'}}\\
  \tbody{\texttt{connect-src 'self'} (model, IP context)}\\
  \tbody{\texttt{worker-src 'self'}}\\[3pt]
  \tbody{partitioned storage: coin secret $k_s$, memo}\\
  \tbody{table, cached model; reads device signals}\\[3pt]
  \tbody{applies memoised \krr{} noise}\\
  \tbody{validates T2 output against schema}\\[3pt]
  \tsub{small, auditable, rarely changes;}\\
  \tsub{the only tier with network access}
};
\node[tier, dashed] (t2) at (12.1,0) {%
  \thdr{T2 \,$\cdot$\, Sealed worker}\\[1pt]
  \tsub{script from vendor $\cdot$ policy from its response}\\[3pt]
  \tbody{\texttt{default-src 'none';}}\\
  \tbody{\texttt{script-src 'wasm-unsafe-eval'}}\\[3pt]
  \tbody{Wasm model + feature fusion}\\
  \tbody{inputs: features, IP context, model bytes}\\
  \tbody{(\texttt{postMessage} only)}\\[3pt]
  \tbody{output: argmax per axis + confidence bucket,}\\
  \tbody{once per epoch}\\[3pt]
  \twarn{no fetch, no import, no navigation}\\
  \tsub{opaque, frequently updated}
};
% arrows at fixed depths below the (aligned) top edges, so they stay horizontal whatever the box heights
\draw[arr] ($(t0.north east)+(0,-1.8)$) -- node[above, font=\scriptsize, text=black!60] {ctx} ($(t1.north west)+(0,-1.8)$);
\draw[arr] ($(t1.north west)+(0,-3.0)$) -- node[below, font=\scriptsize, text=black!60] {label} ($(t0.north east)+(0,-3.0)$);
\draw[arr] ($(t1.north east)+(0,-1.8)$) -- node[above, font=\scriptsize, text=black!60] {feat} ($(t2.north west)+(0,-1.8)$);
\draw[arr] ($(t2.north west)+(0,-3.0)$) -- node[below, font=\scriptsize, text=black!60] {$\hat{y}$} ($(t1.north east)+(0,-3.0)$);
\end{tikzpicture}
\caption{The three tiers of \sys{}. Egress exists only from T0 (the bid request) and T1 (model and IP-context fetch to the vendor origin). T2 has none by construction. The noise is applied in T1, not in the model.}
\label{fig:arch}
\end{figure*}

\sys{} splits the vendor's code into three tiers (Fig.~\ref{fig:arch}) so that the tier that changes often has no network, and the tiers with network are small enough to audit and pin. \textbf{T0, the collector,} is a Prebid.js Real-Time Data submodule~\cite{prebid-rtd} in the publisher page, plain JavaScript pinned through the publisher's Prebid build and Subresource Integrity. It reads the consent state, the publisher's \sharedid{}~\cite{prebid-sharedid} if present, page context and low-entropy user-agent hints, mounts T1, and once per auction merges a validated set of cohort segments into \texttt{ortb2Fragments.global.user.data[]}. \textbf{T1, the bridge,} is a cross-origin iframe on the vendor's origin, served with the vendor's own headers (\texttt{default-src 'none'; script-src 'self' 'wasm-unsafe-eval'; connect-src 'self'; worker-src 'self'; frame-ancestors *}, and a partitioned cookie). It fetches the hash-pinned model into the partitioned Cache API, receives the IP-context vector inline in its own document response, reads the device signals a cross-origin frame can read, holds the per-partition coin secret and memo table, spawns T2, applies the randomised response to T2's output, validates the result against a whitelist of segment identifiers, and forwards only the noised label to T0 by \texttt{postMessage}. \textbf{T2, the sealed worker,} is a dedicated worker whose script is served with \texttt{default-src 'none'; script-src 'wasm-unsafe-eval'}. It compiles the model, fuses features, and returns one arg-max value per axis with a coarse confidence bucket. Inputs arrive as transferred \texttt{ArrayBuffer}s and its only output channel is \texttt{postMessage} back to T1.

\subsection{Why a navigated cross-origin iframe is necessary}
\label{sec:necessity}

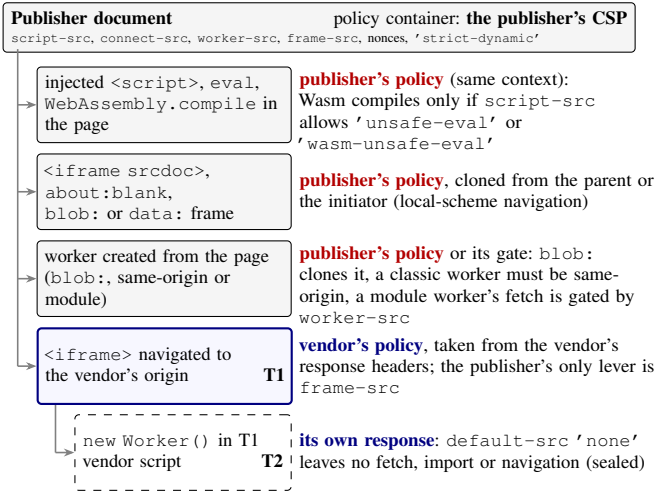
\begin{figure}[t]
\centering
\begin{tikzpicture}[
  font=\scriptsize,
  row/.style={anchor=north west, align=left, minimum height=1.0cm},
  ctx/.style={row, draw, rounded corners=1.5pt, inner sep=3pt, text width=3.15cm},
  pub/.style={ctx, fill=black!4},
  ven/.style={ctx, draw=blue!50!black, thick, fill=blue!3},
  seal/.style={ctx, dashed, fill=white},
  pol/.style={row, inner sep=2pt, text width=4.7cm},
  tree/.style={black!50, semithick},
  edge/.style={-{Stealth[length=1.4mm]}, black!50, semithick}
]
\newcommand{\ppub}{\textcolor{red!70!black}{\textbf{publisher's policy}}}
\newcommand{\pven}[1]{\textcolor{blue!50!black}{\textbf{#1}}}
\node[draw, rounded corners=1.5pt, fill=black!4, anchor=north west, align=left, inner sep=3pt, text width=8.15cm] (root) at (0,0)
  {\textbf{Publisher document}\hfill policy container: \textbf{the publisher's CSP}\\[-1pt]
   {\tiny \texttt{script-src}, \texttt{connect-src}, \texttt{worker-src}, \texttt{frame-src}, nonces, \texttt{'strict-dynamic'}}};
% rows: what third-party script in that document can open, and where each context's policy comes from
\node[pub] (r1) at (0.45,-0.85) {injected \texttt{<script>}, \texttt{eval},\\ \texttt{WebAssembly.compile} in the page};
\node[pol] (p1) at (3.85,-0.85) {\ppub{} (same context): Wasm compiles only if \texttt{script-src} allows \texttt{'unsafe-eval'} or \texttt{'wasm-unsafe-eval'}};
\node[pub] (r2) at (0.45,-2.0) {\texttt{<iframe srcdoc>}, \texttt{about:blank},\\ \texttt{blob:} or \texttt{data:} frame};
\node[pol] (p2) at (3.85,-2.0) {\ppub{}, cloned from the parent or the initiator (local-scheme navigation)};
\node[pub] (r3) at (0.45,-3.15) {worker created from the page\\ (\texttt{blob:}, same-origin or module)};
\node[pol] (p3) at (3.85,-3.15) {\ppub{} or its gate: \texttt{blob:} clones it, a classic worker must be same-origin, a module worker's fetch is gated by \texttt{worker-src}};
\node[ven] (r4) at (0.45,-4.3) {\texttt{<iframe>} navigated to\\ the vendor's origin \hfill \textbf{T1}};
\node[pol] (p4) at (3.85,-4.3) {\pven{vendor's policy}, taken from the vendor's response headers; the publisher's only lever is \texttt{frame-src}};
\node[seal, text width=2.65cm] (r5) at (0.95,-5.45) {\texttt{new Worker()} in T1\\ vendor script \hfill \textbf{T2}};
\node[pol] (p5) at (3.85,-5.45) {\pven{its own response}: \texttt{default-src 'none'} leaves no fetch, import or navigation (sealed)};
% tree lines
\draw[tree] (0.2,-0.62) -- (0.2,0 |- r4.west);
\foreach \r in {r1,r2,r3,r4} { \draw[edge] (0.2,0 |- \r.west) -- (\r.west); }
\draw[tree] (0.7,0 |- r4.south) -- (0.7,0 |- r5.west);
\draw[edge] (0.7,0 |- r5.west) -- (r5.west);
\end{tikzpicture}
\caption{Who authors the CSP of each execution context that third-party script can open from a publisher page, under HTML's policy-container rules~\cite{html-policy-container,html-workers}. Every local-scheme frame and every worker created from the page keeps the publisher's policy or is gated by it. Only a navigated cross-origin iframe takes its policy from the third party's own response (Proposition~\ref{prop:iframe}); a worker created inside that frame takes its policy from its own script response, which is how T2 is sealed.}
\label{fig:policy}
\end{figure}

SDK authors know as folklore that ``the iframe trick works where Wasm is blocked''. It has a precise form (Fig.~\ref{fig:policy}).

\begin{proposition}[Policy independence]
\label{prop:iframe}
Under HTML and CSP Level 3, third-party script in a publisher document can obtain an execution context whose Content Security Policy is authored by the third party if and only if it navigates a child navigable to a non-local URL on an origin it controls.
\end{proposition}
\begin{IEEEproof}[Proof sketch]
HTML's ``determine navigation params policy container'' clones the parent's or initiator's policy container for \texttt{about:srcdoc} and every local-scheme URL and takes the response's policy container only for a network navigation~\cite{html-policy-container}, while classic workers must be same-origin and module workers are governed by the publisher's \texttt{worker-src} chain~\cite{html-workers,csp3}. Full proof in the ancillary \texttt{proofs.md}.
\end{IEEEproof}

The relevant publisher directive is therefore \texttt{frame-src}, the one publishers routinely open because ad creatives render in frames (Section~\ref{sec:crawl} measures how often). WebAssembly is only the most visible consequence. The vendor frame also escapes the publisher's \texttt{connect-src}, nonces and \texttt{'strict-dynamic'}, which is why we call the property policy independence rather than ``Wasm support''.

\subsection{Zero egress by construction}
\label{sec:egress}
CSP's own specification notes that a policy lacking \texttt{default-src} ``cannot mitigate exfiltration''~\cite{csp3}. We ask a narrower question: can a context be built with no network channel at all, so that a model running in it need not be trusted? For a document, no: \texttt{navigate-to} was removed from CSP Level 3, and nothing governs a frame assigning its own \texttt{location}, which is a request. For a dedicated worker, yes. A worker's policy container is initialised from its script's response~\cite{html-workers}, so a worker served with \texttt{default-src 'none'; script-src 'wasm-unsafe-eval'} has no \texttt{fetch}, XMLHttpRequest, WebSocket, EventSource or beacon (all fall back to \texttt{default-src}), no \texttt{importScripts} or dynamic \texttt{import()}, no navigation (no assignable location), and no WebRTC (\texttt{RTCPeerConnection} is not exposed in worker scopes). Every input after the initial load arrives by \texttt{postMessage}. Timing side channels remain, as in any sandbox, but they reach a remote party only through a cooperating networked context, and the only one is the pinned T1. The cost of a worker is that it has no DOM, so screen geometry and similar signals are read in the auditable T1 and passed in. The privacy-critical operation, fusing the feature vector with the IP context and running the model, happens in T2, which has no network.

\subsection{Trusted base and the covert-channel bound}
\label{sec:covert}
T0 and T1 are a few hundred lines that change rarely and can be reviewed and pinned by hash. The model is megabytes, updated weekly, compiled, and \sys{} does not ask anyone to trust it. Its only output is a label that T1 validates against a fixed schema and emits at most once per site per epoch, and because T1 rather than T2 applies the randomised response, the model cannot bypass the noise.

\begin{proposition}[Bounded covert channel]
\label{prop:covert}
Let the taxonomy have axes $a=1..A$ with $C_a$ values, and let T1 emit each axis through \krr{} with keep probability $p_a$, at most once per epoch per site. Then any T2, however malicious, transmits at most
$B=\sum_{a}\big(\log_2 C_a - H(p_a, q_a,\ldots,q_a)\big)$ bits per site per epoch to any observer of the bid stream, where $q_a=(1-p_a)/(C_a-1)$.
\end{proposition}
\begin{IEEEproof}
T1 maps T2's only output through a symmetric channel of capacity $\log_2 C_a - H(\text{row})$ once per epoch, and capacities of independent parallel channels add.
\end{IEEEproof}

At $\varepsilon=3$ per axis, $B \approx 1.24+1.34+1.45+1.26 \approx 5.3$ bits per site per week (Table~\ref{tab:params}). An 18-bit browser fingerprint~\cite{eckersley2010unique} would take more than three weeks to exfiltrate through the label on one site, by which time the memoised coin has held the emitted value constant and partitioned storage has kept the model from seeing the same user elsewhere. No prior client-side advertising system offers a quantified bound of this kind.

\subsection{Budgets and failure modes}
\label{sec:budgets}
Chrome's Heavy Ad Intervention unloads an ad-tagged frame that exceeds 4\,MiB of network bytes (threshold noised upward by up to 1,303\,KiB), 15\,s of CPU in any 30\,s window, or 60\,s of CPU in total, and a vendor frame must assume it is tagged~\cite{chrome-heavyad}. Timers in hidden cross-origin frames are throttled to one wake-up per second or per minute~\cite{chrome-throttling}, and WebAssembly threads need a cross-origin-isolated ancestor chain that fewer than 1\% of pages provide~\cite{almanac2025security}, so T2 runs single-threaded with SIMD. \sys{} therefore budgets the model at $\le 1$\,MiB on the wire, runs inference only while the frame is visible and at background priority, caches the result in T1's partitioned storage, and serves Prebid's auction from the memo table within the RTD module's \texttt{auctionDelay} (typically 50--200\,ms). Storage is partitioned by top-level site in every engine~\cite{chrome-storage-partitioning} and cross-site cookies need the \texttt{Partitioned} attribute~\cite{chips}, so \sys{} is a per-site design by construction, which is also its privacy argument (G1). A publisher whose \texttt{frame-src} excludes the vendor renders the SDK inert. Safari caps script-written storage at seven days, so the coin secret may reset weekly unless a server-set partitioned cookie holds it. Firefox's tracking protection blocks frames from Disconnect-listed domains, so a vendor that behaves like a tracker disappears there. \texttt{getHighEntropyValues()} works in cross-origin frames in Chromium~\cite{ua-client-hints}, but Firefox and Safari expose no hints and clamp or omit several signals, so the model must be evaluated per engine.

\section{Org-Conditional Inference}
\label{sec:inference}

\textbf{Signals.}
\label{sec:signals}
T1 reads what a cross-origin frame can legitimately read: structured user-agent hints~\cite{ua-client-hints}, screen geometry and pixel ratio, logical cores and on Chromium the memory bucket, the WebGL renderer class where exposed, languages, time zone, local hour-of-week, touch support, colour scheme, and per-site session state from its partitioned storage. T0 contributes page context and the consent state. The feature vector is intentionally the one fingerprinting research has characterised~\cite{eckersley2010unique,gomezboix2018hiding,laperdrix2020survey}, and Section~\ref{sec:linkage} explains why using it for a coarse label does not recreate a fingerprint.

\textbf{The IP-context vector.}
\label{sec:ipctx}
The IP address is the one signal that is not on-device, since every web server sees it. \sys{} makes its role explicit and coarse. A server (the vendor's, or the publisher's own) computes from public registry data (RIR delegations, CAIDA AS2Org, PeeringDB, reverse DNS) and an AS-type classifier in the style of ASdb~\cite{cai2010as2org,caida-as2org,ziv2021asdb,ipinfo-company} a vector $x_{\mathrm{ip}} = (\text{egress class},\ \text{size bucket},\ \boldsymbol{\pi}_{\mathrm{ind}},\ \text{confidence})$, where the egress class is one of corporate, residential, cellular, hosting, SASE-shared, relay or VPN, the size bucket is one of LinkedIn's nine staff-count ranges when an organisation resolves and $\bot$ otherwise, and $\boldsymbol{\pi}_{\mathrm{ind}}$ is a prior over top-level industry groups. The vector is delivered inline in T1's document response and never contains the address, a named organisation, or any table derived from member data. The DSP cannot recompute it, because bid requests often carry truncated addresses~\cite{google-ab-openrtb}, so the cohort must be inferred before the request exists.

\textbf{Taxonomy.}
\label{sec:taxonomy}
Cohorts are expressed on four axes that map one-to-one onto the facets a B2B DSP already targets~\cite{linkedin-facets}: seniority $S$ (LinkedIn's ten levels collapsed to five bands: individual contributor, manager, director, VP, executive/owner), function $F$ (26 job functions collapsed to ten groups), organisation size $Z$ (nine \texttt{staffCountRange} buckets collapsed to four: 1--50, 51--500, 501--5,000, 5,001+), and industry $I$ (about twenty top-level groups of LinkedIn's Industry Codes V2). The collapse is not cosmetic: randomised response over the full product of cells destroys utility (Section~\ref{sec:params}), and fine cells are unsafe in small organisations. Titles, which are what buyers type into a campaign, are far too fine to emit, since a title within a resolved organisation is often unique. Every title projects onto an $(S,F)$ cell, and a DSP that targets a title already resolves it that way internally, so \sys{} emits the projection and does not attempt title inference.

\textbf{Model.}
\label{sec:model}
The model factorises the organisation level from the person level,
$p(S,F \mid x) = \sum_{o} p(o \mid x_{\mathrm{ip}})\; p(S,F \mid o, x_{\mathrm{dev}}, x_{\mathrm{ctx}})$,
where $o$ ranges over organisation classes (egress class $\times$ size bucket $\times$ industry group). The org head is largely a lookup on $x_{\mathrm{ip}}$. The person head is a small gradient-boosted ensemble or multilayer perceptron, quantised to 8 bits and under 500\,KB, with per-axis outputs calibrated by temperature scaling or isotonic regression so that the confidence bucket has a meaning a buyer can price. Size and industry come from the org head when an organisation resolves and from the person head's marginal otherwise. When the organisation resolves, its size and industry carry most of the information about seniority and function. When it does not, the person head works from weak device and timing priors and its calibrated confidence should say so.

\textbf{Org-conditional suppression.}
\label{sec:suppression}
Global $k$-anonymity is the wrong guarantee here, because the bid-stream observer sees the label together with an IP that may already resolve the employer. A seniority$\times$function cell populated by thousands of people globally may hold one person inside a 40-employee company.

\begin{definition}[Org-conditional suppression]
\label{def:suppression}
Let $n_{\min}(o)$ be the lower bound of the resolved organisation's size bucket (a large constant when no organisation resolves) and $\pi_a(v \mid o)$ the model's prior for value $v$ of axis $a$ within organisation class $o$. The admissible set of axis $a$ in class $o$ is $U_a(o) = \{ v : n_{\min}(o)\,\pi_a(v \mid o) \ge k \}$. If $|U_a(o)| \le 1$ the axis is not emitted for anyone in $o$. Otherwise it is emitted for everyone in $o$ through randomised response over $U_a(o)$, a true value outside $U_a(o)$ being first replaced by a uniformly random element of $U_a(o)$. The decision depends only on $(o, a)$, never on the person's own value.
\end{definition}

Two details were found by simulation (Section~\ref{sec:linkage}). The admissible set must be decided per organisation class rather than per person, because an emission that depends on the person's own value (``emit $\bot$ if my cell is small'') turns $\bot$ into a disclosure. And the randomised response must range over the admissible set rather than the full value set, because otherwise a flipped coin can land on a value nobody else in the organisation emits, and the noise creates unique records instead of removing them. With $k=30$, the floor LinkedIn applies before reporting any segment of a website's visitors~\cite{linkedin-insights}, the rule means that on the corporate address of a 40-person firm the system emits industry only, while on the same person's home connection it may emit all four axes, because the residential address already maps to a household and the label adds no organisation-level re-identification.

\textbf{Training and model privacy.}
\label{sec:training}
Labels for the person head come from a logged-in professional surface where a member's title maps to seniority and function and the same feature vector is observable on first-party traffic. The org head is trained from public registry data and company web text, where industry classification reaches roughly 88\% on 13 classes~\cite{jagric2024industry}. Because the model ships to every device it is public by construction~\cite{sun2021mind}, so it must not memorise members. Small tabular models memorise little~\cite{carlini2023quantifying}, but rare (organisation, title) combinations are the records at risk, so the recipe trains with DP-SGD~\cite{abadi2016dpsgd} and reports membership inference at low false-positive rates~\cite{carlini2022lira} against the shipped artifact. No member-derived lookup table ever ships client-side.

\section{The Privacy Mechanism}
\label{sec:privacy}

A cohort label emitted in the clear is an attribute disclosure, and one emitted with fresh noise on every request is an averaging exercise. \sys{} emits each axis through a $k$-ary randomised response~\cite{warner1965rr,kairouz2016discrete} whose coin is memoised in the sense of RAPPOR's permanent randomised response~\cite{erlingsson2014rappor} and keyed to the publisher's own first-party identifier. The second choice is the unusual one, and it is what makes the linkability argument clean.

\textbf{Mechanism.}
\label{sec:mechanism}
Fix an axis $a$, its admissible set $U_a$ for the current organisation class ($U_a = V_a$ when no organisation resolves), $C_a = |U_a|$, and budget $\varepsilon_a$. Let $p_a = e^{\varepsilon_a}/(e^{\varepsilon_a} + C_a - 1)$ and $q_a = (1 - p_a)/(C_a - 1)$. T1 holds a random secret $k_s$ in its partitioned storage for the current top-level site $s$ and reads the publisher's first-party identifier $\mathit{id}$ (the \sharedid{}, a per-domain random UUID with no cross-site component~\cite{prebid-sharedid}). For a true value $v$ from T2 (replaced by a uniformly random element of $U_a$ if $v \notin U_a$, using the same memoised coin), T1 computes $r = \mathrm{PRF}_{k_s}(\mathit{id} \,\|\, a \,\|\, v) \in [0,1)$ and emits $\Ltilde = v$ if $r < p_a$, and otherwise the $\lfloor (r - p_a)/q_a \rfloor$-th element of $U_a \setminus \{v\}$ in a fixed order. If the publisher exposes no identifier, the user has opted out of it, or consent for profiling is absent, T1 emits nothing. A user who has declined an identifier has declined cohorts.

\textbf{Properties.} Full proofs are in the ancillary \texttt{proofs.md}.
\begin{proposition}[Local differential privacy per emitted value]
\label{prop:ldp}
Modelling $\mathrm{PRF}_{k_s}$ as a random function, for any true values $v, v' \in V_a$ (admissible or not) and any output $w$, $\Pr[\Ltilde = w \mid v] \le e^{\varepsilon_a} \Pr[\Ltilde = w \mid v']$.
\end{proposition}
\begin{IEEEproof}
For admissible $v$ the output is $p_a$ on $v$ and $q_a$ elsewhere with $p_a/q_a = e^{\varepsilon_a}$, for inadmissible $v$ it is uniform, and the cross-case ratios $p_a C_a$ and $1/(q_a C_a)$ are both at most $e^{\varepsilon_a}$.
\end{IEEEproof}

\begin{proposition}[Anonymity floor within an organisation]
\label{prop:floor}
Let an organisation of $n \ge n_{\min}(o)$ employees have within-organisation prior $\pi_a$, and let every employee emit axis $a$ under Definition~\ref{def:suppression} with independent memoised coins. For every admissible $w$, the expected number of employees emitting $w$ is at least $k\, p_a$.
\end{proposition}
\begin{IEEEproof}
An employee with true value $w$ emits $w$ with probability $p_a$, and $n\,\pi_a(w) \ge k$ by admissibility.
\end{IEEEproof}

\begin{proposition}[No averaging]
\label{prop:memo}
For fixed $(s, \mathit{id}, a, v)$ the emitted value is constant across requests. The privacy loss an observer of site $s$ accumulates under one identifier is at most $\varepsilon_a$ per distinct true value of axis $a$ observed under that identifier, independent of the number of requests or epochs.
\end{proposition}
\begin{IEEEproof}
The output is a deterministic function of $(k_s, \mathit{id}, a, v)$, so only a change of $v$ yields a fresh $\varepsilon_a$-LDP release, as in RAPPOR~\cite{erlingsson2014rappor}.
\end{IEEEproof}

\begin{proposition}[No incremental linkability]
\label{prop:link}
For an observer of site $s$ deciding whether two requests come from the same device: (i) if the requests carry the same identifier, their labels are identical by construction and add nothing to the identifier, and (ii) if they carry different identifiers, because the identifier was reset or the devices differ, their labels are independent $\varepsilon_a$-LDP releases of the respective true values, as if from different sites. The label cannot bridge an identifier reset.
\end{proposition}
\begin{IEEEproof}
(i) is Proposition~\ref{prop:memo}. For (ii) the PRF input differs in $\mathit{id}$, so the coins are independent given the true values.
\end{IEEEproof}
The publisher already decides, through its consent flow and cookie lifetimes, how long its identifier lives. The cohort label inherits that decision with no second knob.

\begin{proposition}[Cross-site linkage advantage]
\label{prop:advantage}
Let two sites (or two identifiers) observe independent releases $\Ltilde_1, \Ltilde_2$ of axis $a$, with population prior $\boldsymbol{\pi}$ over $V_a$. The advantage of the equality test ``$\Ltilde_1 = \Ltilde_2$'' in distinguishing same-device from different-device pairs is $(p_a - q_a)^2 \big(1 - \sum_i \pi_i^2\big)$.
\end{proposition}
\begin{IEEEproof}
Same-device equality has probability $p_a^2 + (C_a - 1) q_a^2$, different-device equality with $v \ne v'$ has $2 p_a q_a + (C_a - 2) q_a^2$, and different devices share a value with probability $\sum_i \pi_i^2$.
\end{IEEEproof}

At $\varepsilon_a = 3$ and a uniform prior over five seniority bands the advantage is $0.628 \times 0.8 \approx 0.50$ (Table~\ref{tab:params}). This is not small, since any consistently released attribute is a linkage bit. The claim is not that the label is unlinkable but that it is no more linkable than an independent LDP release of the attribute, that it adds nothing within an identifier's lifetime, and that the same observer already holds roughly 18 bits of fingerprint entropy in the user agent and IP it is sent~\cite{eckersley2010unique,gomezboix2018hiding}.

\textbf{The weak classifier as a privacy budget.}
\label{sec:weak}
The observer sees a noisy prediction of the attribute, not the attribute. The information carried by one emission is $I(T;\Ltilde)$, computable from the confusion matrix and the \krr{} channel. For a classifier with accuracy $a$ and unstructured errors on $C$ classes, $H(\hat{Y} \mid T) = H_b(a) + (1-a)\log_2 (C-1)$, so a 60\%-accurate five-class seniority model leaks $\log_2 5 - 1.77 \approx 0.55$ bits per emission before randomisation, 0.36 bits at $\varepsilon = 3$ and 0.19 bits at $\varepsilon = 2$. Accuracy alone bounds nothing, because a model wrong in a fixed pattern is a bijection and leaks $\log_2 C$ bits, so the quantity to report is $I(T;\hat{Y})$ from the empirical confusion matrix. The bound is per emission, and Proposition~\ref{prop:memo} is what stops emissions from adding up. Within those cautions, a mediocre, calibrated model is a privacy feature, and Section~\ref{sec:predictable} suggests that is the model one will get.

\textbf{Parameters.}
\label{sec:params}
Noise is applied per axis, not to the product cell. At $\varepsilon = 3$, \krr{} over the $5 \times 10 \times 4 \times 20 = 4{,}000$ joint cells keeps the true cell 0.5\% of the time, while per-axis noise keeps seniority 83\%, function 69\%, size 87\% and industry 51\% of the time (Table~\ref{tab:params}). Under basic composition the joint record is released with $\varepsilon \approx 12$, which we label as weak formal privacy. The protections that carry the design are coarseness, org-conditional suppression, the absence of any identifier, memoisation, and the covert-channel bound. Randomised response is invertible in aggregate: a buyer observing an empirical frequency $y$ estimates the true frequency as $\hat{\pi} = (y - q_a)/(p_a - q_a)$, which is why the noise parameters ship as segment metadata (Section~\ref{sec:integration}). Topics' 5\% random topic over about 350 values was an effective $\varepsilon \approx 8.8$, stabilised per site and epoch, with a witness filter that leaked bits across sites~\cite{thomson2023topics}. \sys{} uses a stronger per-axis budget, memoises per value, and has no witness filter.

\section{Identity-Free Supply--Demand Integration}
\label{sec:integration}

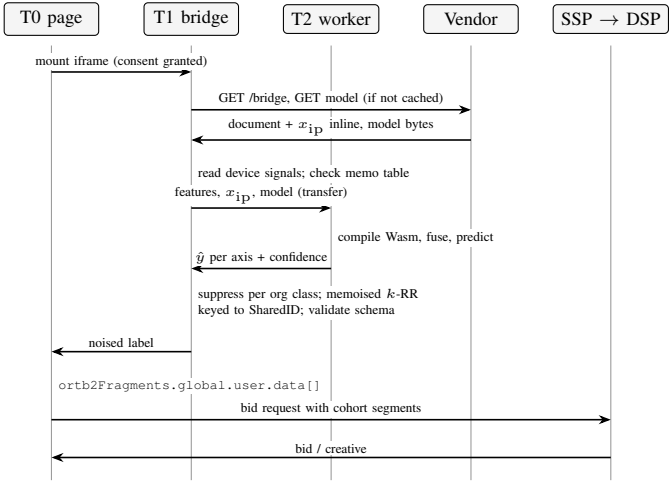
\begin{figure}[t]
\centering
\begin{tikzpicture}[
  font=\scriptsize,
  head/.style={draw, rounded corners=1.5pt, fill=black!4, inner sep=3pt, minimum width=1.25cm, align=center},
  msg/.style={-{Stealth[length=1.6mm]}, semithick},
  self/.style={-{Stealth[length=1.6mm]}, semithick, loop},
  note/.style={font=\tiny, align=left, inner sep=1pt, fill=white},
  lbl/.style={font=\tiny, fill=white, inner sep=0.8pt}
]
\def\xa{0}\def\xb{1.85}\def\xc{3.7}\def\xd{5.55}\def\xe{7.4}
\foreach \x/\n/\t in {\xa/a/T0 page, \xb/b/T1 bridge, \xc/c/T2 worker, \xd/d/Vendor, \xe/e/SSP $\to$ DSP} {
  \node[head] (h\n) at (\x,0) {\t};
  \draw[black!40] (\x,-0.35) -- (\x,-6.1);
}
\draw[msg] (\xa,-0.7) -- node[above=0.6pt, lbl] {mount iframe (consent granted)} (\xb,-0.7);
\draw[msg] (\xb,-1.2) -- node[above=0.6pt, lbl] {GET /bridge, GET model (if not cached)} (\xd,-1.2);
\draw[msg] (\xd,-1.6) -- node[above=0.6pt, lbl] {document + $x_{\mathrm{ip}}$ inline, model bytes} (\xb,-1.6);
\node[note, anchor=west] at (\xb+0.05,-2.05) {read device signals; check memo table};
\draw[msg] (\xb,-2.5) -- node[above=0.6pt, lbl] {features, $x_{\mathrm{ip}}$, model (transfer)} (\xc,-2.5);
\node[note, anchor=west] at (\xc+0.05,-2.9) {compile Wasm, fuse, predict};
\draw[msg] (\xc,-3.3) -- node[above=0.6pt, lbl] {$\hat{y}$ per axis + confidence} (\xb,-3.3);
\node[note, anchor=west] at (\xb+0.05,-3.75) {suppress per org class; memoised \krr{}\\ keyed to \sharedid{}; validate schema};
\draw[msg] (\xb,-4.4) -- node[above=0.6pt, lbl] {noised label} (\xa,-4.4);
\node[note, anchor=west] at (\xa+0.05,-4.85) {\texttt{ortb2Fragments.global.user.data[]}};
\draw[msg] (\xa,-5.3) -- node[above=0.6pt, lbl] {bid request with cohort segments} (\xe,-5.3);
\draw[msg] (\xe,-5.8) -- node[above=0.6pt, lbl] {bid / creative} (\xa,-5.8);
\end{tikzpicture}
\caption{One auction. Network requests leave the browser only from T1 (to the vendor origin) and from T0 (the bid request). T2 receives everything by \texttt{postMessage} and answers the same way. On later auctions within the epoch, T1 serves the label from its memo table without running T2.}
\label{fig:sequence}
\end{figure}

\textbf{Supply side.}
\label{sec:sspside}
T0's \texttt{getBidRequestData} hook merges segments into Prebid's global first-party-data fragment~\cite{prebid-rtd}, which Prebid.js and Prebid Server forward as OpenRTB \texttt{user.data} objects~\cite{openrtb26} (Listing~\ref{lst:userdata}). Fig.~\ref{fig:sequence} shows the sequence within one auction. The taxonomy is registered as a vendor \texttt{segtax} identifier~\cite{iab-segtax}, because IAB's Audience Taxonomy 1.1 has only thirteen ``Employment Role'' nodes and no seniority, organisation size or function~\cite{iab-audience-taxonomy}. The \texttt{segclass} field carries the model version, so a buyer attaches a quality estimate to a model rather than to a vendor, and an \texttt{ext.noise} object states the mechanism and per-axis parameters, so a buyer can invert the randomisation in aggregate. Seller Defined Audiences carry neither provenance nor invertibility, which is part of why they were found non-compliant at scale~\cite{iab2023sdaguide}. Google Ad Manager accepts publisher-provided signals only in IAB taxonomies~\cite{gam-pps}, so cohorts are projected lossily onto Audience 1.1 identifiers for that path, and a DSP that has not implemented the taxonomy can take a cohort package through a curated deal identifier. \sys{} adds nothing to \texttt{user.eids}.

\begin{lstlisting}[float=tb,caption={Cohort segments as emitted into the OpenRTB request. Axis values follow LinkedIn's facet codes and $\bot$ axes are omitted. \texttt{segtax} \texttt{5xx} stands for a vendor-range taxonomy identifier, to be registered with the IAB Tech Lab.},label={lst:userdata}]
"user": { "data": [ {
  "name": "sif.vendor.example",
  "ext": { "segtax": 5xx, "segclass": "m3",
           "noise": { "mech": "krr", "eps": [3,3,3,3],
                      "c": [5,10,4,20] } },
  "segment": [ {"id":"s:6"}, {"id":"f:13"},
               {"id":"z:1001-5000"}, {"id":"i:4"} ]
} ] }
\end{lstlisting}

\textbf{Demand side.}
\label{sec:dspside}
Because the axes map one-to-one onto \texttt{urn:li:seniority}, \texttt{urn:li:function}, \texttt{urn:li:staffCountRange} and \texttt{urn:li:industry}~\cite{linkedin-facets}, a B2B DSP consumes the cohort as its campaign manager already expresses targeting. Attribution uses the identifier the DSP already has: with enhanced conversion tracking, LinkedIn appends a click identifier, \lifatid{}, to the landing-page URL, the advertiser's Insight Tag stores it as a first-party cookie for thirty days, and conversions carry it~\cite{linkedin-lifatid}. \sys{} adds one server-side record at click time, $(\lifatid{} \mapsto \text{bid-time cohort})$, and attributes conversions to cohorts, released with Laplace noise and the floors the platform already applies to audience reporting (300 members in total, 30 per segment~\cite{linkedin-insights}). \sharedid{} and \lifatid{} live on different sites under different partitions and are never joined.

\textbf{Cross-side calibration.}
\label{sec:calibration}
Running \sys{} a second time inside the advertiser's Insight Tag produces, at conversion, a label noised with the advertiser site's own identifier and secret, independent of the impression-side label by Proposition~\ref{prop:link}. Let $Q$ be the composite channel from true value to emitted value (model confusion followed by \krr{}) and $\boldsymbol{\pi}$ the population prior. The agreement matrix $A$ between impression-side and conversion-side labels satisfies $\mathbb{E}[A] = Q^{\top}\,\mathrm{diag}(\boldsymbol{\pi})\,Q$. With the \krr{} parameters known from the segment metadata, $\mathrm{tr}(A)$ lower-bounds the model's self-consistency without ground truth, and the full matrix identifies the Gram form of the confusion matrix: enough to score a source, not enough to recover it. A platform that owns a label source (member profiles) and a measurement surface (its tag on advertiser sites) measures \sys{}'s precision where it has truth and uses the agreement statistic everywhere else. A per-publisher quality score then converts SDA's self-attestation into a statistical audit with no per-user data, a cohort-level conversion rate becomes a bidding signal, and a cookieless successor to website-demographics reporting works on browsers where the platform's cookies are dead.

\section{Measurement and Analysis}
\label{sec:eval}

\subsection{CSP deployment against WebAssembly and third-party frames}
\label{sec:crawl}
Proposition~\ref{prop:iframe} says where \sys{} can run. This measurement says where it does. Existing CSP studies report directive prevalence but not the two quantities that decide deployment: whether a page-context WebAssembly instantiation would fail, and whether a vendor frame would be permitted. Our crawler (ancillary \texttt{csp\_crawl.py}) visits each site's front page in headless Chromium, records enforced policies from headers and \texttt{<meta>} elements, classifies each policy statically (script-governing directive, \texttt{'unsafe-eval'} or \texttt{'wasm-unsafe-eval'}, whether the \texttt{frame-src} $\to$ \texttt{child-src} $\to$ \texttt{default-src} chain admits a probe origin, and whether the \texttt{worker-src} chain does), and probes dynamically by instantiating a minimal WebAssembly module in the page and recording whether a \texttt{CompileError} results. Two populations were crawled on 12 September 2026 with Playwright 1.62 and Chromium 151: the Tranco top 10,000~\cite{lepochat2019tranco} (list 38KVL) and a curated list of 454 business, technology and news publishers where B2B inventory concentrates (ancillary \texttt{b2b\_publishers.txt}). An apex name without an address record was retried once under \texttt{www.}, which recovered 251 Tranco sites. 7,969 Tranco sites and 431 publishers were reachable. The remaining Tranco sites failed on name resolution (1,292), a 25-second timeout (491), or a TLS, HTTP or connection error.

Table~\ref{tab:crawl} gives the result. Page-context WebAssembly fails to compile on 3.7\% of Tranco sites and 4.9\% of the publishers, while the vendor frame is permitted on 88.8\% and 86.8\%. The two cells that decide deployment are small: \sys{} rescues 0.8\% of Tranco sites and 1.6\% of publishers (Wasm blocked, frame permitted) and is inert on 2.9\% and 3.2\% (both blocked), so the frame recovers roughly a fifth and a third, respectively, of the sites on which page-context Wasm fails. The static classifier and the dynamic probe disagree on 0.2\% of Tranco sites (12 of 7,969) and 0.5\% of publishers (2 of 431), and in every such case the page had navigated, reloaded or destroyed its execution context between the recorded response and the probe, or had disabled \texttt{eval} for injected code.

\begin{table}[t]
\centering\footnotesize\setlength{\tabcolsep}{4pt}
\begin{threeparttable}
\caption{CSP deployment relevant to \sys{}, as a percentage of reachable sites.}
\label{tab:crawl}
\begin{tabular}{@{}>{\raggedright\arraybackslash\hangindent=1em\hangafter=1}p{5.3cm}rr@{}}
\toprule
 & Tranco 10k & B2B/news \\
\midrule
Reachable sites ($n$) & 7,969 & 431 \\
Any enforced CSP & 33.7\% & 52.7\% \\
\texttt{script-src} or \texttt{default-src} present & 17.4\% & 24.4\% \\
\quad lacking \texttt{unsafe-eval} and \texttt{wasm-unsafe-eval} & 3.8\% & 5.3\% \\
Page-context Wasm blocked (dynamic probe) & 3.7\% & 4.9\% \\
Vendor frame permitted (\texttt{frame-src} chain) & 88.8\% & 86.8\% \\
Cross-origin module worker permitted & 85.9\% & 81.7\% \\
Wasm blocked, frame permitted (\sys{} rescues) & 0.8\% & 1.6\% \\
Wasm blocked, frame blocked (\sys{} inert) & 2.9\% & 3.2\% \\
\texttt{wasm-unsafe-eval} present & 1.1\% & 1.6\% \\
\bottomrule
\end{tabular}
\begin{tablenotes}\scriptsize
\item Web Almanac prior~\cite{almanac2024security,almanac2025security}: CSP on 19--22\% of pages, \texttt{script-src} in $\sim$17\% of policies, \texttt{unsafe-eval} in 77--78\% of those, and on the order of 1\% of pages blocking Wasm, with no figure for premium publishers or \texttt{wasm-unsafe-eval}. The publisher list is higher on every row.
\end{tablenotes}
\end{threeparttable}
\end{table}

\subsection{Budgets and privacy--utility in closed form}
\label{sec:analytic}
The budgets of Section~\ref{sec:budgets} are satisfiable with margin: a 100-tree gradient-boosted ensemble of depth 6 occupies roughly 75\,KB at 8-bit leaves, a two-layer perceptron with 256 hidden units is under 200\,KB, and both evaluate in well under a millisecond in single-threaded Wasm, whereas the ONNX Runtime Web binary exceeds 11\,MB~\cite{ort-web}, so \sys{} uses a minimal hand-written runtime for tree ensembles and dense layers. Table~\ref{tab:params} tabulates the mechanism's operating points. A buyer de-biasing cohort frequencies incurs variance inflated by $(p_a - q_a)^{-2}$, $1.6\times$ at $\varepsilon=3$ for five values and $3.2\times$ at $\varepsilon=2$, so cohort statistics remain usable at programmatic volumes while per-impression precision is reduced by exactly $1 - p_a$.

\begin{table}[t]
\centering\scriptsize\setlength{\tabcolsep}{2.4pt}
\caption{Operating points of the per-axis mechanism. $p$: probability of emitting the true value. cap.: covert-channel capacity per emission (bits). adv.: linkage advantage under a uniform prior.}
\label{tab:params}
\begin{tabular}{@{}lrrrrrrrrr@{}}
\toprule
 & \multicolumn{3}{c}{$\varepsilon=2$} & \multicolumn{3}{c}{$\varepsilon=3$} & \multicolumn{3}{c}{$\varepsilon=4$} \\
\cmidrule(lr){2-4}\cmidrule(lr){5-7}\cmidrule(lr){8-10}
Axis ($C$) & $p$ & cap. & adv. & $p$ & cap. & adv. & $p$ & cap. & adv. \\
\midrule
Size (4)      & .71 & .68 & .28 & .87 & 1.24 & .51 & .95 & 1.62 & .65 \\
Seniority (5) & .65 & .68 & .25 & .83 & 1.34 & .50 & .93 & 1.83 & .67 \\
Function (10) & .45 & .59 & .14 & .69 & 1.45 & .39 & .86 & 2.29 & .64 \\
Industry (20) & .28 & .41 & .06 & .51 & 1.26 & .23 & .74 & 2.40 & .50 \\
\bottomrule
\end{tabular}
\end{table}

\subsection{The label against the fingerprint, and inside the organisation}
\label{sec:linkage}
The observer receives, with every request, an IP address and a user agent whose combined entropy fingerprinting studies put at roughly 18 bits for desktop browsers~\cite{eckersley2010unique}, reduced but not eliminated by user-agent reduction and partitioning~\cite{gomezboix2018hiding,laperdrix2020survey}. A four-axis label adds at most $\log_2(5\cdot 10\cdot 4\cdot 20) \approx 12$ bits before noise and about 5 bits of channel capacity after it. Those bits are also correlated with what the IP already reveals when an organisation resolves. The relevant anonymity set is not the population but the employees behind one corporate address.

We quantify that case in closed form (ancillary \texttt{reid\_sim.py}). A person with true values $(s, f)$ in an organisation of $n$ employees is unique if no coworker would emit the same observable. With memoised, independent coins this probability is $\sum_w \Pr[w \mid s, f]\,(1 - m(w))^{n-1}$, where $m(w)$ is the probability that a random coworker emits $w$. We draw 150,000 employees with organisation sizes following the U.S. employment distribution by enterprise size (16.2\% in firms under 20 employees, 16.2\% in 20--99, 13.5\% in 100--499, 54.1\% in 500 or more~\cite{susb2022}), log-uniform sizes within each class, and assumed within-organisation priors for seniority ($0.75, 0.13, 0.07, 0.03, 0.02$) and function (a skewed ten-group prior). Results are insensitive to flattening either prior.

\begin{figure}[t]
\centering
\includegraphics[width=\columnwidth]{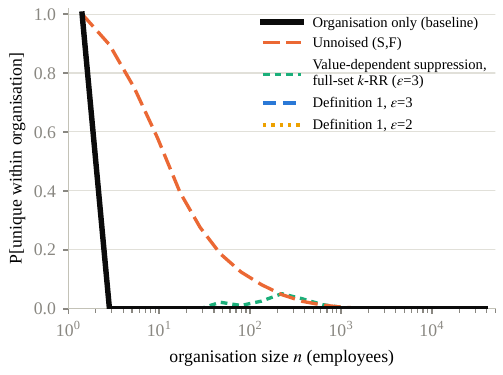}
\caption{Probability that an employee is unique within a resolved organisation, by organisation size. Unnoised seniority$\times$function labels single out most employees of small firms. Value-dependent suppression followed by randomised response over the full value set leaves a bump in mid-sized firms, because flipped coins land on values nobody else emits. Definition~\ref{def:suppression} lies on the organisation-only baseline at both budgets.}
\label{fig:reid}
\end{figure}

Fig.~\ref{fig:reid} gives the result. With the organisation known and no label, 3.7\% of employees are unique: those in single-person firms. Emitting seniority and function in the clear makes 16.3\% unique overall and 75.8\% in firms under 20 employees, and puts 37.8\% of employees in anonymity sets smaller than ten. Suppressing by the person's own cell followed by \krr{} over the full value set yields 4.5\% unique, with a bump of 3.6\% in firms of 100--499: a flipped coin lands on a value nobody else in the organisation emits, and the noise manufactures unique records. Definition~\ref{def:suppression} returns exactly the baseline, 3.7\% unique at both $\varepsilon = 3$ and $\varepsilon = 2$ and 12.3\% in sets smaller than ten, the same as for the organisation alone, with zero excess uniqueness in every size class above single-person firms. The label adds no singling-out capability to what the IP already gives an observer, by construction (Proposition~\ref{prop:floor}) rather than by tuning. The simulation forces a distinction: randomised response protects the attribute (Proposition~\ref{prop:ldp}) but does not by itself protect against singling out, and can worsen it. Suppression does, and its decision must be independent of the person it protects. A full evaluation against real traffic in the framework of Carey et al.~\cite{carey2023reid}, where the observer also holds the fingerprint bits, is future work.

\subsection{What is predictable, and from what}
\label{sec:predictable}
The published evidence per attribute is consistent. A named organisation is inferable deterministically for 10--25\% of U.S. B2B human sessions, only on corporate egress, and vendor rates of 30--65\% require cookies or identity graphs~\cite{demandbase2020wfh,leadfeeder-anon,warmly-matchrates}. Organisation type is inferable for essentially all traffic from AS type and IP-to-company data, with high accuracy for hosting, education and government and ambiguity between business and ISP for small static lines~\cite{ziv2021asdb,ipinfo-company}. Industry follows from the resolved organisation at about 88\% on 13 classes~\cite{jagric2024industry}, and a size-class AUC of about 0.7 is plausible from allocation type alone. For job function there is no passive-signal evidence, only the text-based 52.7\% on nine occupational groups~\cite{preotiuc2015occupation,khaokaew2024workr}. For seniority there is no published evidence at all, and management is 7.2\% of U.S. employment~\cite{bls2025oes}. The strongest device-signal analogue, an Apple device predicting income quartile at 69\% on a balanced binary task~\cite{bertrand2018coming}, is a 1.4$\times$ lift over chance from one binary signal. Our prior for the person head is therefore an AUC of 0.55--0.65 for ``manager or above versus individual contributor'' when no organisation resolves, rising with the organisation's size prior when one does, and chance-level at finer grain. A negative result would be a privacy-positive finding: the device does not leak seniority. No public dataset pairs passive web signals with ground-truth role, so the person head will be evaluated in a consented panel (target $n = 2{,}000$--$5{,}000$). The panel will capture the signals of Section~\ref{sec:signals} together with self-reported seniority, function, organisation size, industry and current network. It will report accuracy, AUC, calibration and $I(T;\hat{Y})$ per axis by signal group and egress class, and cohort assignment rates by age, gender and race proxies. Results will appear in a revision of this preprint.

\section{Legal and Ethical Analysis}
\label{sec:legal}
Local computation does not escape the consent requirement of Article~5(3) of the ePrivacy Directive. The EDPB's Guidelines~2/2023 bring information ``produced locally'' within scope once ``this information or any derivation of this information'' is made available to a third party, including JavaScript-driven collection, protocol headers and IP addresses (\P\P54--56)~\cite{edpb2023guidelines}, and the UK's amended regulation~6 of PECR, in force from February 2026, covers ``information automatically emitted by the terminal equipment''~\cite{pecr-reg6,ico2026storage}. A noised cohort label is such a derivation, so \sys{} is consent-gated: T0 runs only with the Transparency and Consent Framework signals for Purpose~1, Purposes~3 and~4, Purpose~7 and Special Feature~2 where actively requested client hints feed the model~\cite{tcf-policies}, and treats Global Privacy Control as an opt-out, which the CCPA and the \emph{Sephora} action require for a cohort communicated to an ad network~\cite{ccpa-140,sephora2022,cppa2025}. What on-device buys is not exemption but minimisation and verifiability: the consented processing is confined to a context whose outputs are bounded (Propositions~\ref{prop:covert}--\ref{prop:advantage}) and whose networked code is pinned. Inferring seniority or function is profiling under GDPR Article~4(4)~\cite{gdpr}. Ordinary ad selection does not trigger Article~22~\cite{wp251}, but under \emph{SCHUFA} a probability value can itself be a decision when a third party draws strongly on it~\cite{cjeu-schufa}, and seniority correlates with age, so a cohort trained on device price tier, operating-system age and geolocation is a proxy model even with no protected attribute as input~\cite{ali2019discrimination,imana2021auditing}. Following the 2019 Facebook and 2022 Meta settlements~\cite{fb2019settlement,doj2022meta}, \sys{} marks the seniority and function axes as ineligible for housing, employment and credit campaigns in the segment metadata, keeps every axis coarse enough that no cell approximates a protected class, and will report cohort assignment rates by demographic proxies in the panel study. The self-describing segment is the parameter disclosure the Digital Services Act requires~\cite{dsa}. The panel study is consented and reports only aggregates, and the crawl reads public front pages and stores no page content.

\section{Limitations}
\label{sec:limits}

\sys{} does not hide the IP address from the servers it contacts, because no web system can. It confines what is derived from it, computes the derivation once and coarsely, and forbids it from naming an organisation. The signal is also eroding: iCloud Private Relay hides the address for its opt-in users~\cite{sattler2022relay}, SASE egress hides the enterprise~\cite{zscaler-transactions}, and carrier-grade NAT hides everyone on mobile~\cite{richter2016cgn}, so \sys{} degrades to organisation type where these apply. The person head may be weak: if the panel study returns chance-level results, the system reduces to organisation-level cohorts with an explicit confidence, which is still more than the identity-free market has today, and the privacy analysis becomes stronger. Formal privacy is weak by design, since $\varepsilon \approx 12$ across four axes is not a strong guarantee. The guarantees that carry the design are structural (no identifier, confinement, suppression, memoisation, the covert-channel bound), each stated as a proposition with an explicit adversary, and a deployment that wants stronger formal privacy trades utility along Table~\ref{tab:params}. A bot that spoofs device features can enter a high-value cohort, and a publisher can inject segments it did not compute. The server-side IP context is the harder half to spoof, and the cross-side agreement statistic is the buyer's audit against inflation, but neither defends against residential-proxy traffic, which invalid-traffic detection must handle upstream. A browser or extension that blocks the vendor's frame disables \sys{}, and we regard that as correct: user agency belongs in the user agent. The feature vector, storage lifetimes and worker policy enforcement differ by engine, and the model will likely be weaker on the browsers where identity matching is already dead, which is where it is needed. We did not attempt mobile applications and connected TV, where no iframe exists, a pricing model for cohorts, or jurisdictions beyond the EU, UK and United States.

\section{Conclusion}
\label{sec:conclusion}
The browser-native path to privacy-preserving audience cohorts closed in 2025. What remains is ordinary web content, and this paper has argued that ordinary web content is enough. A navigated cross-origin iframe gives a vendor a policy it controls on any publisher, a worker served with \texttt{default-src 'none'} gives the model no network, and a small pinned bridge between them applies a memoised randomised response whose coin rotates with the publisher's own identifier. The result carries no identity, adds no linkability beyond the identifier already present, leaks at most a few bits per site per week even if the model turns hostile, and suppresses where reverse-IP vendors are most confident. Whether the device leaks seniority is a question we have declined to answer by assertion. The evidence says it mostly does not, and for this purpose a weak classifier means strong privacy.

\raggedbottom
\section*{Availability}
The ancillary files contain the full proofs of Propositions~\ref{prop:iframe}--\ref{prop:advantage} (\texttt{proofs.md}), the CSP/WebAssembly crawler of Section~\ref{sec:crawl} with its curated publisher list and the per-site crawl results behind Table~\ref{tab:crawl}, the re-identification simulation, a reference implementation of the \krr{} mechanism with memoisation and aggregate de-biasing, and the response headers and skeleton code for the three \sys{} tiers.

\section*{Acknowledgment}
Large-language-model assistants were used to draft and edit text, to write the crawler, simulation and analysis code, and to check references. The authors reviewed all of the content and take full responsibility for it.

\balance
\bibliographystyle{IEEEtranN}
\bibliography{refs}
\end{document}